\documentclass[12pt,a4paper,english,reqno]{amsart}
\usepackage[T1]{fontenc}
\usepackage[latin9]{inputenc}
\usepackage{amsthm}
\usepackage{amssymb}
\usepackage{tikz}
\usetikzlibrary{arrows,chains,matrix,positioning,scopes}

\makeatletter
\tikzset{join/.code=\tikzset{after node path={%
\ifx\tikzchainprevious\pgfutil@empty\else(\tikzchainprevious)%
edge[every join]#1(\tikzchaincurrent)\fi}}}
\makeatother

\tikzset{>=stealth',every on chain/.append style={join},
         every join/.style={->}}

\makeatletter

\pdfpageheight\paperheight
\pdfpagewidth\paperwidth

\numberwithin{equation}{section}
\numberwithin{figure}{section}

\usepackage{amsfonts}

\usepackage{amsthm}     
\usepackage{amscd}      
\usepackage{euscript}
\usepackage[matrix,arrow,curve]{xy}
\usepackage{pb-diagram}

{
       \newtheorem{theorem}{Theorem}[section]
       \newtheorem{proposition}[theorem]{Proposition}
       \newtheorem{lemma}[theorem]{Lemma}

{\theoremstyle{definition}

       \newtheorem{definition}[theorem]{Definition}
       
}

\newcommand{\pk}{\mathbb{CP}^{k}}

\newcommand{\W}[2]{{H^{#1,#2}}}

\newcommand{\Wnorm}[3]{\left\Vert{#1}\right\Vert_{H^{#2,#3}}}

\newcommand{\Wnormloc}[4]{\left\Vert{#1}\right\Vert_{H^{#2,#3}(#4)}}

\newcommand{\p}{\partial}

\newcommand{\Hol}{\mathcal{H}_{r,k}}
\newcommand{\cHol}{\overline{\mathcal{H}_{r,k}}}

\newcommand{\sym}[1]{Sym^{#1}\Sigma}
\newcommand{\symr}{Sym^r \Sigma}
\newcommand{\xs}{x_s^R}
\newcommand{\ys}{y_s^R}

\newcommand{\varphisa}{\tilde{\varphi}_s^\alpha}
\newcommand{\varphisb}{\tilde{\varphi}_s^\beta}
\newcommand{\vsa}{\tilde{v}_s^\alpha}
\newcommand{\vsb}{\tilde{v}_s^\beta}

\makeatother

\usepackage{babel}

\usepackage{babel}
\begin{document}

\title[ $L^2$ Volume of $\Hol$]{The $L^2$ Volume of the Space of Holomorphic Maps from K\"ahler Riemann Surfaces to $\pk$}
\author{Chih-Chung Liu}
\maketitle

\begin{abstract}

We prove the conjectural formula for the $L^2$ volume of the space of degree $r$ holomorphic maps from a compact K\"ahler Riemann surface of genus $b$ to $\pk$. This formula was posed in \cite{Ba} and rigorously verified in \cite{Sp} for a special case using independent techniques.

\end{abstract}

\section{Introduction}

Given a closed Riemann surface $\Sigma$ of genus $b$ with compatible Riemannian metric with K\"ahler form $\omega$, we study the space $\Hol$, or the space of degree $r$ holomorphic maps from $\Sigma$ to $\pk$. This space has generated great interest in topological field theory, in which $\Hol$ is the moduli space of charge $r$ $\pk$ lumps on $\Sigma$, or the static solitons to field equations of $\pk$ model on $\Sigma$ within the topological sector $r$. From the variational point of view, holomorphic maps are precisely the minimizers to the $L^2$ energy functional

\[E(\phi)=\int_\Sigma \|d\phi\|^2 dvol_\Sigma,\]

\noindent where the norm above will be defined in section 2. The volume of $\Hol$ with respect to the $L^2$ metric defined by the energy functional is a vital quantity in studying the dynamics of solitons. With a clear understanding of the volume growth, the equations of classical gas of lumps may be deduced and motions of solitons at low energy may be approximated. See \cite{M}, \cite{S}, or \cite{Sp} for further details.

A plausible formula for the $L^2$ volume of $\Hol$ has been conjectured in \cite{Ba} and rigorously verified in \cite{Sp} for the case $(r,b)=(1,0)$ (ie. degree 1 maps on $\mathbb{S}^2$). In this paper, we confirm the validity of the formula in general.

\begin{theorem}[Main Theorem]

The $L^2$ volume of $\Hol$ is

\[\frac{(k+1)^b}{q!},\]

\noindent where $q=b+(k+1)(r+1-b)-1$ and $b$ is the genus of $\Sigma$.

\label{Main Theorem Intro}

\end{theorem}

Throughout this paper, we make the universal assumption $r>2-2b$ that ensures $\Hol$ to be a nonempty complex manifold.

\section{Preliminaries}

In this section we aim to provide a concrete description of the space $\Hol$ as a complex manifold of finite dimension, equipped with the standard $L^2$ metrics on maps. Theorem \ref{description of Hol}, the main result of this section, is an easy consequence of results from \cite{Mi}, \cite{Gr}, and \cite{KM}. To ensure that the $L^2$ metric is meaningful, we assume that the images of the holomorphic maps discussed here are non-degenerate, that is, they do not lie in any hyperplane.

For clarity, we first list several basic and well known definitions of the objects used in this article and their basic properties. Readers familiar with these notions may skip to Theorem \ref{description of Hol} and consequences following it. The elementary building block for various moduli spaces is the $r$ symmetric product of $\Sigma$. For Riemann surface $\Sigma$, it is equivalently the space of unordered $r$-tuple of points on $\Sigma$, or the space of effective divisors of degree $r$.

\begin{definition}[Symmetric Product]

Given the $r$-fold direct product of $\Sigma$, we define

\[Sym^r \Sigma := \Sigma \times \ldots \times \Sigma / \sim,\]

\noindent where the equivalent relation is given by

\[(p_1,\ldots,p_r)\sim(p_{\sigma(1)},\ldots,p_{\sigma(r)}),\]

\noindent with $\sigma \in S_r$, the symmetric group.

\label{definition of symmetric product}
\end{definition}

Elements of $\symr$ are denoted by $E=\Sigma_{i=1}^r p_i$, where $p_i$'s need not be distinct. It is a standard fact that $\symr$ is a complex manifold of complex dimension $r$. The local holomorphic coordinates are given by elementary symmetric functions. (See, for example, \cite{Gr}).

The space $\symr$ is closely related to the space of nonzero meromorphic functions on $\Sigma$ of degree $r$, or $K^*_r(\Sigma)$. Since $\Sigma$ is compact, any $f \in K_r^*(\Sigma)$ has $r$ poles and $r$ zeros, counting multiplicities. Each meromorphic function is then associated with a pair of disjoint effective divisors of degree $r$:

\[(f)=(f)_0 - (f)_\infty,\]

\noindent where $(f)_0,(f)_\infty \in \symr$. The converse statement is the classical Abel's Theorem, that a pair of disjoint effective divisors of the same degree gives rise to a meromorphic map as above if they have the same value under the Abel-Jacobi map. Precisely (cf. Chapter V of \cite{G}),

\begin{definition}[Abel-Jacobi Map]

Given $\{\omega_i\}_{i=1}^b$, the $\mathbb{C}$-basis of the space of holomorphic one forms on $\Sigma$ ,  we define

\[\mu: Div(\Sigma) \rightarrow J(\Sigma)\]

\noindent as follows. For $E=\Sigma_{i=1}^l p_i \in Div(\Sigma)$, we have

\begin{equation}
\mu(D)=\left(\begin{matrix} \Sigma_{i=1}^l \int_q^{p_i} \omega_1\\ \vdots \\\Sigma_{i=1}^l \int_q^{p_i} \omega_b\end{matrix}\right).
\label{definition of Abel Jacobi map}
\end{equation}

\noindent Here, $q$ is any point outside the support of $D$, and $J(\Sigma)$ is the Jacobian torus of $\Sigma$, or the quotient of $\mathbb{C}^b$ by the $\mathbb{Z}$-span of row vectors of the $b\times 2b$ period matrix of $\Sigma$, $(\omega_i(\gamma_j))$, where $\{\gamma_j\}$ is the basis for $H_1(\Sigma,\mathbb{Z})\simeq \mathbb{Z}^{2b}$.

\label{definition of Abel Jacobi}

\end{definition}

\noindent One notes that the quotient space eliminates the ambiguities of the definition of $\mu$ in \eqref{definition of Abel Jacobi map}.

A few well-known properties of $\mu$ are relevant for the discussions of this section (see \cite{Gr} for details):

\begin{itemize}
  \item If $E \sim E'$, that is, $E-E'=(f)$ for some $f\in K(\Sigma), \text{ then } \mu(E)=\mu(E')$.
  \item $\mu$ is holomorphic on $\symr$.
  \item $\mu^{-1}(\mu(E))=|E|:=\{F \in Div(\Sigma)\;|\;F\sim E\}$.
\end{itemize}

\noindent The third fact in particular implies that the fiber of $\mu$ at each $\mu(E)\in J(\Sigma)$ is the projectivization of the vector space

\[\mathcal{L}(E) := \{f \in K(\Sigma)\;|\;(f)+E \geq 0\}.\]

\noindent For $r>2b-2$, Riemann-Roch Theorem implies that for all $E \in \symr \subset Div(\Sigma)$, the fiber of $\mu$ at $\mu(E)$ is a projective vector space of dimension $r-b$. Therefore, $\symr$ is a vector bundle of rank $r-b$ over the Jacobian torus $J(\Sigma)$ with the projection map $\mu$.

We can now state the sufficient condition for the difference of a pair of divisors to be principal.

\begin{theorem}[Abel's Theorem]
Given a pair $(E,E') \in \sym{l} \times \sym{l}$ with $\mu(E)=\mu(E')$, there exists $f\in K^*_l(\Sigma)$ so that $(f)=E-E'$.
\label{Abel's Theorem}
\end{theorem}

\noindent The map $f$ is unique up to multiplication of a non-vanishing holomorphic map. For compact Riemann surface, it is then unique up to a nonzero multiple.

We are now equipped with sufficient language to describe the space $\Hol$. In \cite{KM}, a complete description was provided, without explicit proof, for the space of based holomorphic maps from $\Sigma$ to $\pk$.

\begin{definition}
Given $p_0 \in \Sigma$ and $q_0 \in \pk$, we define the space of based holomorphic maps, $\Hol^*$, to be the subset of $\Hol$ of all maps sending $p_0$ to $q_0$.
\label{definition of Hol based}
\end{definition}

\begin{proposition}
The space $\Hol^*$ can be identified with an open set subset of $k+1$ direct product of $\symr$ given by

\[H_{r,k}^* := \{(E_0,\ldots,E_{k}) \in \symr \times \cdots \times \symr\;|\;\cap_j E_j = \emptyset \text{ and }\mu(E_0)=\mu(E_j)\;\forall\;j\}.\]

\label{description of based holomorphic maps}
\end{proposition}

\begin{proof}
Given $f \in \Hol^*$, since $f(\Sigma)$ is nondegenerate, it is a curve of degree $k$ in $\pk$. Let $H_1$ be the hyperplane of $\pk$ defined by the canonical section $s_1$ of $\mathcal{O}(1)$. In local coordinate, points of $H_1$ is of the form $[0:z_1:\ldots:z_k]$. Then, counting multiplicities, we have

\[H_1 \bigcap f(\Sigma) = \{q_1,\ldots,q_k\}.\]

\noindent Since $f$ is nondegenerate, $q_1,\ldots,q_k$ are in general position in $\mathbb{CP}^{k-1}$, and we may define coordinate patches of $\pk$ so that each $q_j$ has homogeneous coordinate $[\delta_{ij}]_{i=0}^k$, $1\leq j \leq k$.

For each $j \in \{1,\ldots,k\}$, let $E_j := f^{-1}(q_j)$. Since $f$ is of degree $r$, it follows that $E_j \in \symr$. For $j=0$, we let $E_0:=f^{-1}([1:0:\ldots:0])\in\symr$. We must show that the tuple $(E_0,\ldots,E_k)$ satisfies the condition of the open set given above. It is clear that $E_0 \cap \left(\cap_{j=1}^k E_j\right) = \emptyset$ since $f(E_j) \subset H_1$ for all $j \geq 1$ and therefore $\cap_{j=0}^k E_j = \emptyset$. For the level set conditions, it suffices to observe that all divisors given above are linear equivalent. Indeed, define $\pi_{ji} \in K(\pk)$ by

\[\pi_{ji}([z_0:\ldots:z_k]):= \frac{z_j}{z_i}.\]

\noindent Each $\pi_{ji}$ is globally defined on $\pk$ and therefore $\pi_{ji}\circ f \in K^*(\Sigma)$. Moreover, with the choice of coordinates associated to $q_j$'s above, it is clear that for all $j$,

\[(\pi_{j0}\circ f)=E_0 - E_j.\]

\noindent It then follows that all divisors chosen above are linearly equivalent.

Conversely, given $(E_0,\ldots,E_k)$ satisfying the conditions of the open subset. Abel's Theorem ensures that for each $j\geq 1$, there exists $f_j \in K^*_r(\Sigma)$ such that $(f_j)=E_j-E_0$, unique up to a nonzero multiple $\lambda_j \in \mathbb{C}^*$. Since $f$ is a based map, the constant multiples $\lambda_j$ are completely determined. We define the corresponding map $f$ coordinate-wise by

\[f := [1:f_1:\ldots:f_k].\]

\noindent This map is holomorphic away from $E_0$. However, near each point of $E_0$, we may holomorphically extend $f$ across the singularities by multiplying the poles of minimum order of vanishing (See \cite{Mi} for details), obtaining a holomorphic map from $\Sigma$ to $\pk$. $f$ is clearly of degree $r$, since $f^{-1}([0:1:\ldots :1])=E_0 \in \symr$. The two associations constitute a one-to-one correspondence between $\Hol^*$ and $H_{r,k}^*$.

\end{proof}

Lifting the restriction on based maps corresponds to the freedom of choosing $(\lambda_j)_{j=1}^k \in (\mathbb{C}^*)^{\times k}$ and we have established

\begin{theorem}
$\Hol$ is diffeomorphic to $\Hol^* \times  \pk$, where $\Hol^*$ is defined in Proposition \ref{description of based holomorphic maps}.
\label{description of Hol}
\end{theorem}

\begin{proof}
In the proof of Proposition \ref{description of based holomorphic maps}, the function $f$ defined by $(E_j)\in \Hol^*$ is unique up to a choice of $k+1$ nonzero constants (including 1). Moreover, the choice of these $k+1$ constants define the same $f$ if they lie on the same line in $\mathbb{C}^{k+1}$.
\end{proof}

The complex dimension of $\Hol$ can be readily computed.

\begin{lemma}
$\Hol$ is a complex manifold of complex dimension
\[(k+1)r-k(b-1).\]
\label{dimension of Hol}
\end{lemma}

\begin{proof}
It suffices to show that the space of based maps $\Hol^*$ is of dimension $(k+1)r-kb$, since the remaining summand in the conclusion of Theorem \ref{description of Hol} contributes $k$ to the dimension. Note that $\Hol^*$ is the submanifold of a $(k+1)r$-dimensional manifold defined by the linear map $s: (\symr)^{\times (k+1)} \to (J(\Sigma))^k$ given as,

\[s(E_0,E_1,\ldots,E_k):= (\mu(E_1)-\mu(E_0),\ldots,\mu(E_k)-\mu(E_0)).\]

\noindent It suffices to show that $s$ is of rank $kb$ at every point on the inverse image of $0$. Indeed, consider the linear map $\Delta : \symr \times \symr \to J(\Sigma)$ defined by

\[\Delta(E,E')=\mu(E)-\mu(E').\]

\noindent From the third property of the Abel-Jacobi map listed immediately after Definition \ref{definition of Abel Jacobi}, we see that the inverse image of 0 is precisely

\[\{(E,E')\;|\;E' \in |E|\},\]

\noindent which is of dimension $r+\dim |E|=2r-b$ by Riemann-Roch Theorem (note again $r>2-2b$). The rank of $\Delta$ is then $2r-(2r-b)=b$ and therefore each component of $s$ contributes rank $b$ to the entire map. This completes the proof.

\end{proof}

The space $\Hol$ is not compact and is potentially problematic when discussing convergence of integrals over it. As such, we introduce the well known Uhlenbeck compactification of $\Hol$ and will later discuss the convergence of integrals over this compact set. The following description is summarized from \cite{BDW}:

\begin{definition}
The \emph{Uhlenbeck  compactification} of $\Hol$ is defined by

\[\cHol :=  \bigsqcup_{l=0}^r \left[Sym^l \Sigma \times \mathcal{H}_{r-l,k}\right],\]

\noindent where $\mathcal{H}_{r-l,k}$ is the space of degree $r-l$ holomorphic maps from $\Sigma$ to $\pk$.

\label{definition of Uhlenbeck compactification}
\end{definition}

\noindent The topology of $\cHol$ is given by weak* topology on $Sym^l\Sigma$ and $\mathcal{C}_0^\infty$ topology on $\mathcal{H}_{r-l,k}$. Precisely, for each $(E,f)\in\cHol$, we associate the distribution $e(f)+\delta_E$, where $e(f)=|df|^2$ the energy density and $\delta_E$ is the Dirac distribution supported in $E$. Around $(E,f)$, there are two local bases of neighbourhoods of topologies: one basis $W$ of $e(f)+\delta_E$ in the weak* topology and one basis $N$ of $f$ in $C_0^\infty(\Sigma \ E)$. Collecting them together, we form a basis of neighbourhoods

\[V(E,f)=\{(E',f')\in \cHol\;|\;E'\in W \text{ and }f,\in N \}.\]

\noindent  In \cite{SU}, it was shown that $\cHol$ with this topology is compact. Since $\Hol$ is open and dense subset of $\cHol$, it is indeed a smooth compactification of $\Hol$. In our main result, we will derive the convergence statement of integrals on $\cHol$ and ensure that the complement does not contribute to the integral. One notes that for $l=0$, this topology degenerates into $C_0^\infty$ topology of functions, which coincide with the subspace topology inherited from the direct sums of symmetric products. Indeed, the locations of zeros determines the smooth topology as well as the energy densities. Moreover, the smooth structure of symmetric spaces provides a natural smooth structure for $\cHol$, namely, the $\mathcal{C}^\infty$ topology controlled by the locations of the zeros.

We may describe $\cHol$ with explicit coordinates. We observe that $\cHol$  is a closed submanifold  of $(\symr)^{\times (k+1)} \times \pk$. The component from symmetric products determines the zeros from $k+1$ meromorphic functions on $\Sigma$, a total of $k+1$ unordered $r$-tuples of points from $\Sigma$. With the $k$ constraint equations by Abel-Jacobi maps of rank $b$, it follows that there are $m=r(k+1)-kb$ points from $(\symr)^{\times (k+1)}$ that determine the correct zeros of each meromorphic function up to a free choice of point in $\pk$. $\cHol$ therefore possesses local coordinate description

\begin{equation}
(w_1,\ldots w_m,w_{m+1},\ldots,w_m),
\label{coordinate of Hol}
\end{equation}

\noindent where the coordinates  $(w_1,\ldots w_m)$ is the local coordinate of the submanifold of $(\symr)^{\times (k+1)}$ defined by Abel-Jacobi equations and $(w_{m+1},\ldots,w_m)$ is the local coordinate for $\pk$. Here, $m=(k+1)r-k(b-1)$ as in Lemma \ref{dimension of Hol}. \eqref{coordinate of Hol} is also a local coordinate for the open subset $\Hol$, with the additional open condition that $(w_1,\ldots w_m)$ does not yield meromorphic functions with common zeros or poles.

Lastly, we define the $L^2$ metrics on $\Hol$ with respect to the K\"ahler metric $\omega$ on $\Sigma$ and Fubini-Study metric on $\pk$. Given $f \in \Hol$, the tangent space of $\Hol$ at $f$ can be identified with the space of sections of the pullback bundle of $T \mathbb{CP}^{k-1}$ via $f$:

\[T_f\Hol \simeq \Gamma (f^* T\mathbb{CP}^{k-1}).\]

\noindent Given $u,v \in T_f \Hol$,  they can be viewed as a pullbacked sections on $\Sigma$, which can be pushed forward by $f$ to be tangent vectors on $\mathbb{CP}^{k-1}$, on which Fubini-Study metric can be applied. We define

\begin{definition} [$L^2$ Metrics on $\Hol$]

The $L^2$ metric on $\Hol$ is given by

\begin{equation}
\left<u,v\right>_{L^2}:=\int_\Sigma \left<f_*u,f_*v\right>_{FS}vol_\Sigma.
\label{definition of L 2 metric on maps}
\end{equation}

\noindent Here, the $f_*$ denotes the pushforward of $f$.

\end{definition}

\noindent With respect to this metric, we aim to verify the conjectured formula for the volume of $\Hol$.

\section{Vortex Equations and Its Moduli Spaces}

We hereby provide a brief summary of the $s$-vortex equations and its moduli spaces. Standard texts and papers for this subject include \cite{JT}, \cite{S}, \cite{Ro}, \cite{M} (of more physical interest), \cite{Br}, \cite{Br1}, \cite{G}, and \cite{Ba} (of more mathematical aspects). The conclusion below is drawn directly from \cite{L}, which contains a comprehensive summary the standard texts.

On the same Riemann surface defined above, we consider a Hermitian line bundle $(L,H)$ with degree $r>2-2b$ possessing smooth global sections.  For each $k \in \mathbb{N}$ and $s\in \mathbb{R}$, the $s$-vortex equations are defined on $(D,\phi_0,\ldots,\phi_k) \in \mathcal{A}(H)\times (\Omega^0(L))^{\times(k+1)}$, the product of the space of $H$-unitary connections and $k+1$ tuples of smooth sections, we have the $s-$vortex equations

\begin{equation}
\begin{cases}
F_{D}^{(0,2)}=0\\
D^{(0,1)}\phi=0\\
\sqrt{-1}\Lambda F_{D}+\frac{s^{2}}{2}(\sum_{i=0}^k|\phi_i|^2_H-1)=0.
\end{cases}\label{s-vortex}
\end{equation}

\noindent Here, $F_D$ is the curvature form of the connection $D$. We use $\phi$ to denote a $(k+1)$-tuple of sections $(\phi_i)_{i=0}^k$ and the second equation is an abbreviation that $D^{(0,1)} \phi_i =0 \;\forall i$. Via Bogomolny-type arguments, the set of equations are precisely the minimizing equations the $s$-Yang-Mills-Higgs energy functional:

\begin{equation}
YMH_{\tau.s}(D,\phi):=\frac{1}{s^{2}}\left|\left|F_{D}\right|\right|_{L^{2}}^{2}+\sum_{i=1}^k\left|\left|D\phi_i\right|\right|_{L^{2}}^{2}+\frac{s^{2}}{4}\left|\left|\sum_{i=1}^k|\phi_i|^2_H-1 \right|\right|_{L^{2}}^{2}.\label{YMH}
\end{equation}

The solution space to \eqref{s-vortex} is invariant under unitary gauge action, and we let $\nu_{k+1}(s)$ be the $U(1)$ gauge classes of solutions to \eqref{s-vortex}. Elements of this space are called vortices. For a line bundle $L$, it is well known that the stability condition, or the non-emptyness of $\nu_{k+1}(s)$, is simply

\[s^2 \geq \frac{4\pi r}{\text{vol }\Sigma}.\]

\noindent Note that we intentionally define $k+1$ sections for vortex equations to accommodate the main theme of this article, the space of holomorphic maps to $\pk$. It has been noted in \cite{BDW} and \cite{Ba} that $\Hol$ is diffeomorphic to the open subset of $\nu_{k+1}(s)$, denoted by $\nu_{k+1,0}(s)$, consisting of vortices with $k+1$ having no common zeros:

\[\nu_{k+1,0}(s):= \{[D,(\phi_i)_{i=0}^k]\in \nu_{k+1}(s)\;|\;\cap_i \phi_i^{-1}(0)= \emptyset\}.\]

\noindent We have,

\begin{theorem}
For all $s$ such that $s^2 \in [\frac{4\pi r}{\text{Vol}(\Sigma)}, \infty)$, there is a diffeomorphism

\[\Phi_s: \Hol \to \nu_{k+1,0}(s).\]

\label{moduli space correspondence}
\end{theorem}

\begin{proof}({\em Sketch})
We provide only the definition of $\Phi_s$ as a reference for later construction. For complete proofs and justifications, see \cite{Br} for the case $s=1$ and \cite{L} for generalization to all $s$.

The inverse direction of $\Phi_s$ is more obvious. For $[D,\phi]\in \nu_{k+1,0}(s)$, we define

\[\Phi_s^{-1}([D,\phi])(z):=\tilde{\phi}(z)=[\phi_0(z),\ldots,\phi_k(z)].\]

\noindent The map is well defined as $\phi_i(z)$'s are never all zero. Moreover, on a $U(1)$ bundle, the transition map multiplies each section by a uniform nonzero scalar, having no effect on the definition of $\tilde{\phi}$. Since each component is holomorphic near every $z \in \Sigma$ and has order of vanishing $r$, $\Phi_s^{-1}([D,\phi)]$ is indeed an element of $\Hol$.

The forward direction is also a standard construction in algebraic geometry. We start with a holomorphic map $\tilde{\phi}\in \Hol$. Consider $\mathcal{O}_{\mathbb{CP}^{k}}(1)$, the anti-tautological line bundle over $\pk$ with hyperplane sections $s_0,\ldots,s_k$. Each $s_j$ vanishes precisely on the hyperplane defined by $z_j=0$.
Let $L=\tilde{\phi}^*\mathcal{O}_{\mathbb{CP}^k}(1)$ be the pullbacked line bundle on $\Sigma$ endowed with sections $\phi=(\phi_0,\ldots,\phi_k)\in \Omega^0(L)^{\oplus {(k+1)}}$ by pulling back
$s_1,\ldots,s_k$ via $\tilde{\phi}$. The map $\tilde{\phi}$ also endows a
holomorphic structure $\bar{\partial}_L$ and a background metric $H$ on $L$ by pulling back the standard complex structure and Fubini-Study metric $H_{FS}$ on $\mathcal{O}_{\mathbb{CP}^k}(1)$, determining a corresponding background unitary connection $D$. We then find a suitable complex gauge transformation turning $(D,\phi)$ into a pair that solves the vortex equations \eqref{s-vortex}. On a line bundle, such a transformation is given by a nonzero real smooth function $g_s=e^{u_s}$, acting on $(D,\phi)$ in the standard ways:

\[g_s \cdot D = D\left(e^{u_s} (\bar{\p}_L e^{-u_s})\right):= D_s,\]

\noindent and

\[g_s \cdot \phi_i = e^{u_s} \phi_i := \phi_{i,s}.\]

\noindent Let $\phi_s=(\phi_{i,s})_{i=0}^k$. Plugging $(D_s,\phi_s)$ into \eqref{s-vortex}, we see that the function $u_s$ must satisfies the following Kazdan-Warner type equation:

\begin{equation}
-\Delta_\omega u_s+\frac{s^{2}}{2}\,\sum_{i=0}^k|\phi_i|_{H}^{2}e^{2u_s}+\left(\sqrt{-1}\Lambda F_{H}-\frac{ s^{2}}{2}\right)=0.
\label{s-vortex 1}
\end{equation}

Here, $\Delta_\omega$ is the positive definite Laplacian determined by $\omega$. It has been shown in \cite{Br} (for $s=1$) and \cite{L} (for general $s$) that equation \eqref{s-vortex 1} has a unique smooth solution provided that $s$ is large enough. The analytic tool used to solve $u_s$ is developed in \cite{kw} using the method of super and sub solutions. We will mention more details in the next section. That the two correspondences are smooth inverses to each other is verified in \cite{Br} and \cite{L} via straightforward arguments.

\end{proof}

With this correspondence, the entire vortex moduli space can be identified without difficulty. For a general vortex $[D,\phi] \in \nu_{k+1}(s)$, we first locate the common zeros of the $\phi_i$'s, an element $E \in \sym{l}$ for some $l \in \{0,\ldots,r\}$. The $k+1$ sections define a holomorphic map from $E$ to $\pk$ in identical way as $\Phi_s^{-1}$. This is map is bounded away from $E$, and can therefore be uniquely extended to the entire $\Sigma$ by locally dividing out the zeros of minimum order of vanishing near each common zero. The holomorphic map defined is then of degree $r-l$. See \cite{BDW} and \cite{Mi} for details. We have

\begin{proposition}
The vortex moduli space $\nu_k(s)$ is diffeomorphic to the Uhlenbeck compactification of $\Hol$:

\[\cHol=\bigsqcup_{l=0}^r \left(Sym^l \Sigma \times \mathcal{H}_{r-l,k}\right).\]
\label{moduli space of entire vortex}
\end{proposition}

The $L^2$ metrics for $\nu_{k+1}(s)$ is also defined naturally. At $(D_s,\phi_s) \in \mathcal{A}(H)\times \Omega^0(L)^{\times (k+1)}$, we define,

\begin{definition}[$L^2$ Metrics on $\nu_{k+1}(s)$]
\begin{equation}
g_{s}((\dot{A_{s}},\dot{\phi_{s}}),(\dot{A_{s}},\dot{\phi_{s}}))=\int_{\Sigma}\frac{1}{2s^{2}}\dot{A_s}\wedge
\bar{*}\dot{A_s}+<\dot{\phi_{s}},\dot{\phi_{s}}>_{H}vol_{\Sigma}.
\label{L 2 metric}
\end{equation}
\end{definition}

\noindent Here, $(\dot{A}_s, \dot{\phi}_s)$ denotes a tangent vector in $T_{(D_s,\phi_s)}(\mathcal{A}(H)\times \Omega^0(L)^{\times k+1}) \simeq \mathcal{A}^1(\Sigma) \times \Omega^0(L)^{\times (k+1)}$, and $\left<\dot{\phi}_s,\dot{\phi}_s\right>=\sum_{i=0}^k \left<\dot{\phi}_{i,s},\dot{\phi}_{i,s}\right>$. The identification is justified by the fact that $\Omega^0(L)$ is a vector space and $\mathcal{A}(H)$ is an affine space modeled on the vector space $\Omega^1(\Sigma)$, the space of complex valued one forms on $\Sigma$. (cf. Chapter V of \cite{K}). This identification also justifies the applications of Hodge star $\bar{*}$ and $<\cdot,\cdot>_H$ in the integrand of \eqref{L 2 metric}, since $(\dot{A_s},\dot{\phi_s})$ lies in essentially isomorphic spaces as $(D_s,\phi_s)$ does. By choosing tangent vectors orthogonal to $\mathcal{G}$-gauge transformations, \eqref{L 2 metric} descends to a well defined metric on the quotient space $\left(\mathcal{A}(H) \times \Omega^0(L)^{\oplus (k+1)}\right) / \mathcal{G}$ and restricts to the open subset $\nu_{k+1,0}(s)$ (cf. \cite{L}). With respect to this metric, Baptista has explicitly computed the formula for the volume of $\nu_{k+1}(s)$ in \cite{Ba}, generalized directly from \cite{MN}:

\begin{theorem}
For a degree $r$ Hermitian line bundle $L$ over a closed K\"ahler Riemann surface $\Sigma$ of genus $b$, where $r>2b-2$, the $L^2$ volume of the vortex moduli space $\nu_{k+1}(s)$ is given by

\begin{equation}
Vol \nu_{k+1}(s) = \sum_{i=0}^b \frac{b!(k+1)^{b-i}}{i!(q-i)!(b-i)!}\left(\frac{4\pi}{s^2}\right)^i\left(Vol \Sigma - \frac{4\pi}{s^2}r\right)^{q-i},
\label{volume of moduli space}
\end{equation}

\noindent where

\[q=b+(k+1)(r+1-b)-1.\]

\label{volume for nu k}
\end{theorem}

\noindent The method for the construction of this formula is to realize $\nu_{k+1}(s)$ as a fiber bundle over the moduli space of holomorphic structures, which are identified with the Jacobian torus. The K\"ahler form of the $L^2$ metric \eqref{L 2 metric} can be decomposed into parallel and horizontal components in terms of characteristic classes and fiberwise integration is performed to obtain the formula.

The asymptotic behaviours of $\nu_{k+1}(s)$, \eqref{L 2 metric}, and \eqref{volume of moduli space} as $s\to \infty$ are our main interest. One observes that in the $s-$vortex equations \eqref{s-vortex}, the only $s$ dependence lies in the third equation. As $s$ increase, the curvature term in the third equation becomes negligible compared to the section terms and we may formally define the vortex moduli space at infinity, $\nu_{k+1}(\infty)$, to be the unitary gauge class of solutions to

\begin{equation}
\begin{cases}
F_{D}^{(0,2)}=0\\
D^{(0,1)}\phi=0\\
\sum_{i=0}^k|\phi_i|^2_H-1=0.
\end{cases}\label{infinity-vortex}
\end{equation}

\noindent It is then immediately clear from the third equation in \eqref{infinity-vortex} that no $k+1$ tuples of sections in $\nu_{k+1}(\infty)$ vanish simultaneously. We have $\nu_{k+1}(\infty) = \nu_{k+1,0}(\infty)$. Moreover, \eqref{infinity-vortex} indicates that solutions to these equations consist of precisely pairs of unitary connection $D$ and $k+1$ $D$-holomorphic sections with image in the unit sphere $\mathbb{S}^{2k+1}$. Modulo gauge transformation, each vortex at $s=\infty$ then defines a holomorphic map from $\Sigma$ to $\pk$ in exactly identical way provided in the proof of Theorem \ref{moduli space correspondence}. We have a diffeomorphism

\[\Phi_\infty: \Hol \to \nu_{k+1}(\infty)\]

\noindent defined with precisely the same way as $\Phi_s$. A naturally plausible attempt to compute the $L^2$ volume of $\Hol$, as conjecturized in \cite{Ba}, is then to let $s \to \infty$ in \eqref{volume of moduli space}. With the established result in \cite{L} that the $L^2$ metric \eqref{L 2 metric} converges to the $L^2$ metric \eqref{definition of L 2 metric on maps} in the sense of Cheeger-Gromov, we aim to further establish the fact that formula \eqref{volume of moduli space} is true at $s=\infty$.

\section{Convergence of Moduli Spaces and Their Metrics}

We provide a brief conclusion of solving $u_s$ using techniques developed in \cite{kw}, which are important tools for the main constructions of this article. We recall equation \eqref{s-vortex 1}, the PDE characterizing solutions to vortex equations:

\begin{equation}
-\Delta_\omega u_s+\frac{s^{2}}{2}\,\sum_{i=0}^k|\phi_i|_{H}^{2}e^{2u_s}+\left(\sqrt{-1}\Lambda F_{H}-\frac{ s^{2}}{2}\right)=0.
\label{Kazdan Warner again}
\end{equation}

\noindent Solving this equation is equivalent to solving the Kazdan-Warner equations in \cite{kw} in the following ways. Normalizing the K\"ahler metric $\omega$ so that $Vol_\omega \Sigma =1$, consider

\[
\begin{aligned}c(s): & =2\int_{\Sigma}\left(\sqrt{-1}\Lambda F_{H}-\frac{s^{2}}{2}\right)dvol_\omega\\
 & =2\int_{\Sigma}\sqrt{-1}\Lambda F_{H}\omega- s^{2}=2c_1- s^2,
\end{aligned}
\]

\noindent where $c_{1}=\int_{\Sigma}\sqrt{-1}\Lambda F_{H} \omega$ is independent
of $s$ and $H$. Consider $\psi$, the solution to:

\begin{equation}
\Delta_\omega \psi=\left(\sqrt{-1}\Lambda F_{H}-\frac{s^{2}}{2}\right)-\frac{c(s)}{2}=\sqrt{-1}\Lambda F_{H}-c_{1},
\label{equaiton for psi}
\end{equation}

\noindent which is clearly independent of $s$. Here, $\Delta_\omega$ is the Laplacian operator with respect to the K\"ahler form $\omega$. Setting $\varphi_{s}:=2(u_{s}-\psi)$, one can readily see that $u_s$ satisfies \eqref{Kazdan Warner again} if and only if $\varphi_s$ satisfies:

\begin{equation}
\Delta\varphi_{s}-\frac{s^{2}}{2}\left(\sum_{i=0}^k\left|\phi_i\right|_{H}^{2}e^{2\psi}\right)e^{\varphi_{s}}-c(s)=0.\label{Kazdan Warner Equation 0}
\end{equation}

\noindent This is of the form:

\begin{equation}
\Delta\varphi_{s}=-\left(\frac{s^{2}}{2}h\right)\, e^{\varphi_{s}}+c(s),\label{Kazdan Warner Equation}
\end{equation}

\noindent with the strictly negative norm function

\begin{equation}
h =-\sum_{i=0}^k|\phi_i|_H^2 e^{2\psi}
\label{norm function}
\end{equation}

\noindent and $c(s)<0$ for $s$ in appropriate range. Lemma 9.3 in \cite{kw} guarantees the unique solution to exist using the method of super and sub solutions, a consequence of the maximum principle applied to the operators $L_s:= \Delta_\omega -s^2k I$, where $k>0$ is a constant determined by $h$. In \cite{kw}, sufficient conditions to solve \eqref{Kazdan Warner Equation} are merely $h \leq 0$. However, to obtain further uniformity and convergent behaviors of $u_s$ over $s$ requires $h<0$, which depends crucially on the non-simultaneous vanishing of the $k+1$ sections. In \cite{L}, the author has proved

\begin{theorem}
On a compact Riemannian manifold $M$ without boundary, let $c_1$ be any constant, $c_2$ any positive constant,
and $h$ any negative smooth function. Let $c(s) = c_1 -c_2 s^2$, for each $s$ large enough,
the unique solutions $\varphi_s \in C^\infty$ for the equations

\[
\Delta\varphi_{s}=c(s)-s^2he^{\varphi_s}.
\]

\noindent are uniformly bounded in all Sobolev spaces. Moreover, in the limit
$s \to \infty$, $\varphi_{s}$ converges smoothly to

\[\varphi_\infty = \log\left(\frac{c_2}{-h}\right),\]

\noindent the unique solution to

\[
he^{\varphi_{\infty}}+c_{2}=0.
\]

\label{Main Theorem on Convergence of Gauge}

\end{theorem}

\noindent The theorem is an independent analytic result with the additional benefit that it holds on general compact Riemannian manifolds of any dimension and complex structure is not necessary. In particular, the uniformity of Kazdan-Warner equations can be deduced from this theorem with $c_1=\int_\Sigma \sqrt{-1}\Lambda F_H dvol_\Sigma$ and $c_2=1$. Theorem \ref{Main Theorem on Convergence of Gauge} is proved by explicitly constructing the family of approximated solutions $v_s$ that approach $\varphi_\infty$ smoothly and moreover, the difference of $v_s$ and $\varphi_s$ approaches zero in any Sobolev norm. The family of approximated solutions is given by

\begin{equation}
v_s := \log \left(\frac{\Delta_\omega \left(-\log(-h)\right)-c(s)}{-s^2h}\right)
\label{approximated solutions}
\end{equation}

\noindent with the family of error terms

\begin{equation}
E_s := \Delta_\omega \left(\frac{\Delta_\omega \left(-\log(-h)\right)-c(s)}{s^2}\right)
\label{error terms}
\end{equation}

\noindent so that

\[\Delta_\omega v_s = c(s) =s^2he^{v_s}+E_s.\]

\noindent The two families $v_s$ and $\varphi_s$ satisfy the following lemma:

\begin{lemma}
For all $l \in \mathbb{N}$, the family of actual and approximated solutions to the Kazdan-Warner equations \eqref{Kazdan Warner Equation}, $\varphi_s$ and $v_s$, we have

\[\lim_{s\to \infty} \Wnorm{\varphi_s - v_s}{l}{\infty}=0,\]

\noindent where $H^{l,\infty}$ is the Sobolev space of order $l$ and exponent $\infty$.

\label{Convergence of approximated and actual solutions}
\end{lemma}

\noindent The lemma will again be a crucial part of the main result of this paper. The complete proof is summarized in the next section.

The convergence result and its proof may be directly applied to prove a conjecture on convergence of $L^2$ metrics of vortex moduli spaces \eqref{L 2 metric} to the ordinary $L^2$ metric on maps \eqref{definition of L 2 metric on maps}, posed in \cite{Ba} and proved in \cite{L}. The convergence holds in the sense of Cheeger-Gromov:

\begin{definition} [Cheeger-Gromov Convergence] For all $l \in \mathbb{N}$ and $p \geq 1$, a family of $n$-dimensional Riemannian manifolds
$(M_s,g_s)$ is said to converge to a fixed Riemannian manifold $(M,g)$ in $\W{l}{p}$, in the sense of
Cheeger-Gromov, if there is a covering chart $\{U_k, (x_i^k)\}$ on $M$ and a family of
diffeomorphisms $F_s:M \to M_s$, such that

\begin{equation}
 \Wnormloc{F_s^*(g_s)\left(\frac{\partial}{\partial x_i},\frac{\partial}{\partial x_j}\right) - g\left(\frac{\partial}{\partial x_i},\frac{\partial}{\partial x_j}\right)}{l}{p}{U_k}\to 0.
 \label{Cheeger Gromov convergence}
 \end{equation}

\noindent as $s \to \infty$, for all $k$ and $i,j \in \{1,\ldots,n\}$.
\end{definition}

\noindent We have

\begin{theorem}[Baptista's Conjecture]
Equipping $\pk$ with the Fubini-Study metric, the sequence of metrics $g_s$ on $\nu_{k+1,0}(s)$ given by \eqref{L 2 metric} converges in all Sobolev spaces, in the sense of Cheeger-Gromov, to the $L^2$ metric $<\cdot,\cdot>_{L^2}$ on $\Hol$ given by \eqref{definition of L 2 metric on maps}. The family of diffeomorphisms are precisely $\Phi_s$, as constructed in Theorem \ref{moduli space correspondence}.
\label{Precise Baptista's Conjecture}
\end{theorem}

\begin{proof}(\emph{Sketch})

We provide a brief outline of the proof in \cite{L} with some adjustments to the notations. Fix $\tilde{\phi}\in\Hol$. Pick canonical hyperplane sections $s_0,\ldots,s_k$ with constant Fubini-Study norm 1. On the pullback line bundle $L=\tilde{\phi}^*\mathcal{O}(1)$ with pullback metric $H$ locally given by

\[H(z)=\tilde{\phi}^* H_{FS} (z) = \frac{1}{\sum_{i=0}^k |\tilde{\phi}_i(z)|},\]

\noindent the pullback sections $(\phi_i=\tilde{\phi}^*s_i)_{i=0}^k$ are of constant $H$-norm 1. As in section 3 these initial data give rise to the norm function $h$ and the special gauge function $u_s$ that lead to the solution of vortex equations. Infinitesimally, given a tangent vector $\tilde{\xi}_\alpha \in T_{\tilde{\phi}}\Hol \simeq \Gamma(\tilde{\phi}^* T\pk)$, there are corresponding infinitesimal variations of metric and gauge function, denoted by $H^\alpha$ and $u_s^\alpha$. Moreover, there are $k+1$ sections $\phi^\alpha:=(\phi_i^\alpha)$, identified with $\tilde{\xi}_\alpha$ via the pullback of the Euler sequence on $\pk$ via $\tilde{\phi}$. With these, we identify the pushforward of $\tilde{\xi}_\alpha$ via $\Phi_s$:

\begin{eqnarray}
\Phi_{s,*}(\tilde{\xi}_\alpha) &=& \left(A_s^\alpha,\phi_s^\alpha \right) \nonumber \\
&=& \left(A^\alpha+2 \frac{\p u_s^\alpha}{\p z} dz,
e^{u_s}\phi^\alpha+e^{u_s}u_s^\alpha\phi\right) \in T_{[D_s,\phi_s]}\nu_{k+1,0}(s), \nonumber \\
\label{pushforward definition}
\end{eqnarray}

\noindent where $A^\alpha$ is the initial tangent vector of the connection form induced by $\tilde{\xi}_\alpha$. Given $\tilde{\xi}_\alpha,\tilde{\xi}_\beta \in T_{\tilde{\phi}}\Hol$, the pullback metric $g_s^*$ at $\tilde{\phi}$ are then determined by an $m \times m$ matrix of smooth functions given by

\begin{eqnarray}
g_s^*\left(\tilde{\xi}_\alpha,\tilde{\xi}_\beta\right) &:=& g_s\left(\Phi_{s,*}\left(\tilde{\xi}_\alpha\right),\Phi_{s,*}\left(\tilde{\xi}_\beta\right)\right) \nonumber \\
&=& \int_\Sigma \frac{\left(A^\alpha+2\frac{\p u_s^\alpha}{\p z}\right)\overline{\left(A^\beta+2\frac{\p u_s^\beta}{\p z}\right)}}{2s^2} vol_\Sigma \nonumber \\
& &+ \int_\Sigma \left(\left<\phi^\alpha,\phi^\beta\right>_H e^{2u_s}+(u_s^\alpha u_s^\beta e^{2u_s})\right) vol_\Sigma. \nonumber \\
\label{L 2 metric pull back}
\end{eqnarray}

\noindent A fact implicitly used in this formula is that the sections $\phi$ and the metric $H$ are chosen so that the sum of the $H$ norms of $\phi$'s are constant, and therefore $<\phi,\phi^\alpha>_H=0\;\forall \alpha$.

Recall that $\varphi_s=2(u_s-\psi)$, where $\psi$ is the solution to

\[\Delta_\omega \psi = \sqrt{-1}\Lambda F_H - c_1,\]

\noindent and the background sections are chosen so that $e^{2u_s}=-he^{\varphi_s}$, \eqref{L 2 metric pull back} can be rewritten as

\begin{eqnarray}
& &g_s^*\left(\tilde{\xi}_\alpha,\tilde{\xi}_\beta\right)  \nonumber \\
&=& \int_\Sigma \frac{\left(A^\alpha+\frac{\p \varphi_s^\alpha}{\p z}+2\frac{\p \psi^\alpha}{\p z}\right)\overline{\left(A^\beta+\frac{\p \varphi_s^\beta}{\p z}+2\frac{\p \psi^\beta}{\p z}\right)}}{2s^2} vol_\Sigma \nonumber \\
& &+ \int_\Sigma \left(\left<\phi^\alpha,\phi^\beta\right>_H \left(-he^{2u_s}\right)-(he^{\varphi_s})^\alpha u_s^\beta\right) vol_\Sigma. \nonumber \\
\label{L 2 metric pull back alternative}
\end{eqnarray}

The pointwise definition of $g_s^*$ extends naturally to a smooth function near $\tilde{\phi}$. Take an open neighborhood $\mathcal{U}$ near $\tilde{\phi}$ with local coordinate $\{w_1,\ldots,w_m\}$ centered at $\tilde{\phi}$ as described in \eqref{coordinate of Hol}. Let $\{\tilde{\xi}_i,\ldots,\tilde{\xi}_m\}$ be the corresponding frame of $T\Hol$ over $\mathcal{U}$. Each holomorphic function $\tilde{\eta}\in \mathcal{U}$ yields a unique set of background metric, sections, and gauge function depending smoothly on variation of holomorphic functions. The functions $H,\psi,h,\varphi_s$ and $u_s$ in the integrand of \eqref{L 2 metric pull back alternative} are therefore understood to be smooth functions defined on $\mathcal{U}\times \Sigma$, and we amend these notations to $\tilde{H},\tilde{\psi},\tilde{h},\tilde{\varphi}_s$ and $\tilde{u}_s$ to reflect their dependencies on $w_1,\ldots,w_m$. The formula \eqref{L 2 metric pull back alternative} then defines $m \times m$ smooth functions on $\mathcal{U}$:

\begin{eqnarray}
g_{\alpha,\beta,s}^* &:=& g_{s}^*\left(\tilde{\xi}_\alpha,\tilde{\xi}_\beta\right) \nonumber \\
         &=& \int_\Sigma \frac{\left[\tilde{A}^\alpha+\frac{\p \tilde{\varphi}_s^\alpha}{\p z}+2\frac{\p \tilde{\psi}^\alpha}{\p z}\right]\overline{\left[\tilde{A}^\beta+\frac{\p \tilde{\varphi}_s^\beta}{\p z}+2\frac{\p \tilde{\psi}^\beta}{\p z}\right]}}{2s^2}  dvol_\Sigma \nonumber \\
         & & + \int_\Sigma \left[\left<\phi^\alpha,\phi^\beta\right>_H \left(-he^{\tilde{\varphi}_s}\right)-\left(he^{\tilde{\varphi}_s}\right)^\alpha \tilde{u}_s^\beta \right]dvol_\Sigma.
\label{L 2 metric pull back coordinate}
\end{eqnarray}

\noindent The proof of Theorem \ref{Main Theorem on Convergence of Gauge} is then mimicked to show that the first and third terms in the integrand of \eqref{L 2 metric pull back coordinate} converge to zero and $he^{\varphi_s} \to 1$ smoothly as $s\to \infty$, at all points on $\mathcal{U}$.

\end{proof}

To this end, we have established that $g_s^*$ is locally represented by a finite collection of smooth functions, each converging smoothly, as functions on $\mathcal{U}$, to the finite collection of smooth functions representing the ordinary $L^2$ metrics. To establish the convergence of volumes, we must show that the volume forms induced by each $g_s^*$ converses to the volume form induced by the $L^2$ metric in $L^1(\mathcal{U})$ as $s\to \infty$. An obvious challenge comes from the fact that singularities for $\varphi_s$ may develop near the boundary of the non-compact set $\mathcal{U}$ corresponding to vortices with common zeros. In such situations, pointwise convergence of volume forms does not necessarily imply that integrals, namely the volume, converge to the integral of the limiting volume form. However, the next section assures that the convergence in integral continues to hold.

\section{Main Constructions}

In this section we establish the main theorem of this article.

\begin{theorem}[Main Theorem]

The $L^2$ volume of $\Hol$ is

\[\frac{(k+1)^b}{q!},\]

\noindent where $q=b+(k+1)(r+1-b)-1$ and $b$ is the genus of $\Sigma$.

\label{Main Theorem}

\end{theorem}

\noindent The theorem follows from several technical computations. Since the spaces discussed in the Main Theorem are open, where boundedness conditions required for the convergence of integrals are more difficult to achieve, we wish to establish the result on the compactifications of these spaces, noting that compactifications do not affect the $L^2$ volumes.

\begin{proposition}
With respect to the $L^2$ metric defined in \eqref{L 2 metric}, we have

\[\text{Vol} \nu_{k+1}(s) = \text{Vol} \nu_{k+1,0}(s),\]

\noindent for all $s$ large enough.

\label{vortices with common zero is of zero measure}
\end{proposition}

\begin{proof}
The proposition is an immediate consequence of Propositions \ref{moduli space of entire vortex} and Lemma \ref{dimension of Hol}. We recall the diffeomorphic correspondence

\[\nu_{k+1}(s) \xrightarrow{\sim} \mathcal{H}=\bigsqcup_{l=0}^r Sym^l \Sigma \times Hol_{r-l}(\Sigma,\pk).\]

\noindent Let $\nu_{k+1,l}(s)$ be the class of vortices with common zeros a divisor of degree $l$, which corresponds to the $l^{th}$ stratification above. Lemma \ref{dimension of Hol} then implies that $\nu_{k+1,l}(s)$ is of dimension

\[k(r-b+1)+r-kl,\]

\noindent and therefore $\nu_{k+1,l}(s)$ is of strictly lower dimension than $\nu_{k+1,0}(s)$ for all $l>0$. The arguments are valid for all $s$ large enough, including $\infty$, and the proposition is therefore established. .
\end{proof}

This proposition implies, together with \eqref{volume of moduli space}, that

\begin{equation}
Vol \nu_{k+1,0}(s) = \sum_{i=0}^b \frac{b!(k+1)^{b-i}}{i!(q-i)!(b-i)!}\left(\frac{4\pi}{s^2}\right)^i\left(Vol \Sigma - \frac{4\pi}{s^2}r\right)^{q-i},
\label{volume of moduli space nonvanishing}
\end{equation}

\noindent where $q=b+(k+1)(r+1-b)-1$, which implies that

\begin{equation}
Vol\Hol = Vol\overline{\Hol}.
\label{equality of volume}
\end{equation}

\noindent Therefore, we may compute the volume of the compactified spaces, where $L^p$ convergence is more feasible. To discuss convergence on the compactified domain, we first construct smooth extensions of the basic functions and forms that give rise to metric functions $g_{\alpha,\beta,s}$. Throughout the entire constructions below, we reserve $\alpha,\beta \in \{1,\ldots,m\}$ as indices for coordinates of the finite dimensional manifold $\cHol$ $(w_1,\ldots,w_m)$ described in \eqref{coordinate of Hol}. Moreover, we adopt the notation "$f^\alpha$" to denote the derivative, or the infinitesimal variation of an object $f$, induced from the $\alpha^{th}$ component of the local frame of $T\cHol$.

\begin{proposition}
The curvature $F_{\tilde{H}}$ given by the pullback metric $\tilde{H}$ can be smoothly extended to a family of $(1,1)$ form on $\Sigma$ smoothly parametrized by $\cHol$.
\label{extension of curvature}
\end{proposition}

\begin{proof}
We first extend the domain of the pullback metric $\tilde{H}$ to $\cHol\times\Sigma$. One notes that $\cHol$ is identified with

\begin{equation}
\overline{H_{r,k}} := \{(E_0,\ldots,E_k)\in  \symr \times \cdots \times \symr \;|\;\mu(E_0)=\mu(E_j)\;\forall j\} \times \pk.
\label{identification of cHol}
\end{equation}

\noindent Namely, it is $\Hol$ without the open condition $\cap_j E_j = \emptyset$. Each $(E_0,\ldots,E_k)\in\overline{H_{r,k}}$ is associated to an element in $\cHol$ as follows. Let

\[E:=\cap_j E_j=\sum_{a=1}^{n_E} d_a p_a \text{ with } \sum_{a=1}^{n_E} d_a=l.\]

\noindent Evidently, $E \in Sym^l \Sigma$ for some $l$. For each $j$, let $E_j'=E_j - E$. The tuple $(E_0',\ldots,E_k')$ then have no point in common and still satisfies the defining equations for $\cHol$ (and $\Hol$.) It then defines a holomorphic map $\tilde{\phi}_E:\Sigma \to \pk$ of degree $r-l$ as in Proposition \ref{description of based holomorphic maps}. We then consider the pullback of Fubini-Study metric via $\tilde{\phi}_E$ with the locally defined function:

\begin{equation}
\tilde{H}'(E_0,\ldots,E_k)(z)=\frac{1}{\sum_{i=0}^k |\tilde{\phi}_{E,i}(z)|^2}.
\label{extension of the pull back Fubini Study metric}
\end{equation}

\noindent This metric is defined on the line bundle $L=\tilde{\phi}_E^*(\mathcal{O}(1))$ over $\Sigma$ with degree $r-l$.

The definition of the family $\tilde{H}'$ clearly agrees with $\tilde{H}$ originally defined on $\Hol\times\Sigma$ with $E=\emptyset$. The corresponding family of curvatures forms is then locally given by

\begin{equation}
F_{\tilde{H}'} := \bar{\p}{\p} \log \tilde{H}'.
\label{extension of the family of curvatures}
\end{equation}

\noindent $F_{\tilde{H}'} $ is clearly smooth and we claim that it is a smooth extension of $F_{\tilde{H}}$. Indeed, for $E=\sum_{a=1}^{n_E} n_ap_a$, pick a neighborhood $U$ of $\Sigma$ so that $U \cap E =\{p_a\}$. On $U-\{p_a\}$, we consider tuples of divisors

\[K=(K_0,\ldots,K_k)\]

\noindent without common point that defines holomorphic map $\tilde{\phi}_K$. The family of maps give rise to the family of curvatures

\begin{eqnarray}
F_{\tilde{H}} &=& -\bar{\p}\p \log \left(|z-p_a|^{2d_a}\sum_{i=0}^k \left|\frac{\tilde{\phi}_K(z)}{(z-p_a)^{d_a}}\right|^2\right) \nonumber \\
                &=& -d_a\bar{\p}\p \log |z-p_a|^2  -\bar{\p}\p \log \left(\sum_{i=0}^k \left|\frac{\tilde{\phi}_K(z)}{(z-p_a)^{d_a}}\right|^2\right) \nonumber \\
                &=& -\bar{\p}\p \log \left(\sum_{i=0}^k \left|\frac{\tilde{\phi}_K(z)}{(z-p_a)^{d_a}}\right|^2\right). \nonumber \\
\label{formula for extended curvature}
\end{eqnarray}

\noindent The third equality is true because $z \neq p_a$ and therefore $\bar{\p}\p \log |z-p_a|^2=0$. The third line in \eqref{formula for extended curvature} is defined on the entire $U$, which also agree with the curvature given by the extended metric function $\tilde{H}'$ when $d_a$ points of $K_0,\ldots,K_k$ converge to $p_a$, and the proof is complete.

\end{proof}

\begin{proposition}
The norm function $h$ defined in \eqref{norm function}, amended in Theorem \ref{Precise Baptista's Conjecture}, can be smoothly extended to a negative smooth function on $\cHol\times\Sigma$.
\label{extension of norm function}
\end{proposition}

\begin{proof}

The proof follows easily from Proposition \ref{extension of curvature}. With the definition of pullback Fubini-Study metric in mind, we examine the definition of $\tilde{h}: \cHol\times\Sigma \to \mathbb{R}$

\begin{eqnarray}
\tilde{h}(\tilde{\phi},z)&=&-e^{2\tilde{\psi}}\sum_{i=0}^{k}\left|\tilde{\phi}^*(s_i)\right|_{\tilde{H}}^2  \nonumber \\
                       &=& -e^{2\tilde{\psi}}\frac{\sum_{i=0}^k\left|s_i(\tilde{\phi}(z))\right|^2}{\sum_{j=0}^k\left|\tilde{\phi}(z)\right|^2},   \nonumber \\
\label{definition of extended norm function}
\end{eqnarray}

\noindent where $\tilde{\psi}$ solves the equation

\[\Delta_\omega \tilde{\psi}=\sqrt{-1}\Lambda F_{\tilde{H}}-c_1.\]

\noindent One notes that in the second line of the definition of $\tilde{h}$, the fraction is in fact a well defined smooth function on $\cHol$, as accumulations of common zeros appear simultaneously on the numerator and denominator with the same order. For the exponential term, we observe that since the curvature form can be smoothly extended to $\cHol$, so can $\tilde{\psi}$. Indeed, the solution to the Laplacian equation

\[\Delta_\omega \tilde{\psi}' = \sqrt{-1}\Lambda F_{\tilde{H}'}-\overline{\sqrt{-1}\Lambda F_{\tilde{H}'}}\]

\noindent satisfies the condition and extends the original $\tilde{\psi}$. Therefore, $\tilde{h}$ is smoothly extended to a nonzero smooth function on $\cHol\times\Sigma$.

\end{proof}

Propositions \ref{extension of curvature} and \ref{extension of norm function} allow us to establish the convergence of metric functions of $g_s^*$ to $g_{L^2}$ in all $L^p(\cHol)$.  For each $\alpha,\beta \in \{1,\ldots,m\},$ we re-examine the formula of $g_{\alpha,\beta,s}^*$:

\[g_{\alpha,\beta,s}^* =  X_{\alpha,\beta,s} + Y_{\alpha,\beta,s} + Z_{\alpha,\beta,s} \in \mathcal{C}^\infty (\mathcal{U}),\]

\noindent where

\begin{equation}
X_{\alpha,\beta,s}=\int_\Sigma\frac{\left[\tilde{A}^\alpha+\frac{\p \tilde{\varphi}_s^\alpha}{\p z}+2\frac{\p \tilde{\psi}^\alpha}{\p z}\right]\overline{\left[\tilde{A}^\beta+\frac{\p \tilde{\varphi}_s^\beta}{\p z}+2\frac{\p \tilde{\psi}^\beta}{\p z}\right]}}{2s^2} dvol_\Sigma,
\label{X s}
\end{equation}

\begin{equation}
Y_{\alpha,\beta,s}=-\int_\Sigma \left(\tilde{h}e^{\tilde{\varphi}_s}\right)^\alpha \tilde{u}_s^\beta dvol_\Sigma,
\label{Y s}
\end{equation}

\noindent and

\begin{equation}
Z_{\alpha,\beta,s}=\int_\Sigma \left<\phi^\alpha,\phi^\beta\right>_H \left(-\tilde{h}e^{\tilde{\varphi}_s}\right) dvol_\Sigma.
\label{Z s}
\end{equation}

\noindent We now prove their expected convergence in $L^p(\cHol)$. The domains of all the functions and forms considered below are now defined on $\cHol\times\Sigma$ unless otherwise specified.

\begin{proposition}
For all $\alpha,\beta$, the functions

\[Z_{\alpha,\beta,s}=\int_\Sigma \left<\phi^\alpha,\phi^\beta\right>_H \left(-\tilde{h}e^{\tilde{\varphi}_s}\right) dvol_\Sigma\]

\noindent converge to

\[ \int_\Sigma \left<\phi^\alpha,\phi^\beta\right>_H dvol_\Sigma\]

\noindent as $s\to\infty$ in $L^p(\cHol)$ for all $p$.

\label{L p convergence of Z}
\end{proposition}

\begin{proof}

By the Main Theorem 3.4 in \cite{L}, the function $-\tilde{h}e^{\tilde{\varphi}_s}$ converges to $1$ pointwise on $\Hol$.  Therefore,

\[Z_{\alpha,\beta,s} \to \int_\Sigma \left<\phi^\alpha,\phi^\beta\right>_{\tilde{H}} = \int_\Sigma \left<\tilde{\phi}_*(\tilde{\xi}^\alpha),\tilde{\phi}_*(\tilde{\xi}^\beta)\right>_{H_{FS}},\]

\noindent pointwise as $s\to\infty$ on $\cHol$.

To ensure the convergence in $L^p(\cHol)$, we need appropriate bounds to apply dominated convergence theorem. The bound is obtained from uniform estimates of $\tilde{h}e^{\tilde{\varphi}_s}$ with the approximated solutions to Kazdan-Warner equations, constructed in \cite{L}.

For

\[\tilde{c}(s):=2\int_\Sigma \sqrt{-1}\Lambda F_{\tilde{H}} -s^2\]

\noindent we define the approximated solutions

\begin{equation}
\tilde{v}_s := \log \frac{-\Delta_\omega \left(\log(-\tilde{h})\right) - \tilde{c}(s)}{-s^2\tilde{h}}
\label{definition of approximated solutions}
\end{equation}

\noindent with error functions

\begin{equation}
\tilde{E}_s := \Delta_\omega \log \left(\frac{\Delta\left(-\log(-\tilde{h})\right)-\tilde{c}(s)}{s^2}\right)
\label{definition of error functions}
\end{equation}

\noindent so that

\begin{equation}
\Delta_\omega \tilde{v}_s = -s^2\tilde{h}e^{\tilde{v}_s} + \tilde{E}_s
\label{Kazdan-Warner equation for approximated solutions}
\end{equation}

\noindent By Proposition \ref{extension of norm function}, $\tilde{v}_s$ and $\tilde{E}_s$ are smooth on $\cHol\times\Sigma$.  $\tilde{c}(s)$ is constant on each connected component of $\cHol$ and may be treated as constant when differentiating with respect to any variable. These functions provide natural bounds for $Z_{\alpha,\beta,s}$. Indeed, one observes that

\[\tilde{h}e^{\tilde{v}_s} = \frac{-\Delta_\omega \left(\log(-\tilde{h})\right) - \tilde{c}(s)}{-s^2}\]

\noindent and therefore

\[\int_\Sigma -\tilde{h}e^{\tilde{v}_s} vol_\Sigma = \frac{\tilde{c}(s)}{s^2} \leq K_1, \]

\noindent where $K_1$ is a uniform constant over $\cHol$. Note that we have normalized the K\"ahler form of $\Sigma$ so that $Vol(\Sigma)=1$. The difference between $\tilde{h}e^{\tilde{\varphi}_s}$ and $\tilde{h}e^{\tilde{v}_s}$ is estimated uniformly using the maximum principal of $\Delta_\omega$. For each $s$ and $R\in\cHol$, since $\tilde{\varphi}_s(R,\cdot)-\tilde{v}_s(R,\cdot)$ is smooth on the compact set $\Sigma$, we may pick $\xs ,\ys \in\Sigma$ so that

\[\tilde{\varphi}_s(\xs)-\tilde{v}_s(\xs)=\sup_{z\in\Sigma}\{\tilde{\varphi}_s(z)-\tilde{v}_s(z)\},\]

\noindent and

\[\tilde{\varphi}_s(\ys)-\tilde{v}_s(\ys)=\inf_{z\in\Sigma}\{\tilde{\varphi}_s(z)-\tilde{v}_s(z)\}.\]

\noindent The maximum principle of $\Delta_\omega$ then implies that

\begin{equation}
\Delta_\omega \left(\tilde{\varphi}_s-\tilde{v}_s\right) (\xs) \leq 0,
\label{maximum principle sup}
\end{equation}

\noindent and

\begin{equation}
\Delta_\omega \left(\tilde{\varphi}_s-\tilde{v}_s\right) (\ys) \geq 0.
\label{maximum principle inf}
\end{equation}

\noindent Subtracting \eqref{Kazdan-Warner equation for approximated solutions} from the Kazdan-Warner equations for $\tilde{\varphi}_s$:

\[\Delta_\omega \tilde{\varphi}_s=-s^2\tilde{h}e^{\tilde{\varphi}_s}+\tilde{c}(s),\]

\noindent \eqref{maximum principle sup}, \eqref{maximum principle inf}, and the choices of $\xs$,$\ys$ together then yield the estimates

\begin{equation}
\left(\tilde{h}e^{\tilde{\varphi}_s}-\tilde{h}e^{\tilde{v}_s}\right)(R,z) \leq \frac{\tilde{E}_s}{(-s^2\tilde{h})}\Big|_{(R,\xs)}\left(\tilde{h}e^{\tilde{v}_s}\right)(R,z),
\label{upper bound}
\end{equation}

\noindent and

\begin{equation}
\left(\tilde{h}e^{\tilde{\varphi}_s}-\tilde{h}e^{\tilde{v}_s}\right)(R,z) \geq \frac{\tilde{E}_s}{(-s^2\tilde{h})}\Big|_{(R,\ys)}\left(\tilde{h}e^{\tilde{v}_s}\right)(R,z),
\label{lower bound}
\end{equation}

\noindent for all $R\in \cHol$. (See the proof of Theorem 3.4 in \cite{L} for the complete process.) Therefore, we have

\begin{eqnarray}
\left|\tilde{h}e^{\tilde{\varphi}_s}-\tilde{h}e^{\tilde{v}_s}\right| &\leq& \frac{1}{s^2}\sup_{\cHol\times\Sigma}\left|\tilde{h}e^{\tilde{v}_s}\right|\sup_{\cHol\times\Sigma}\left|\frac{\tilde{E}_s}{\tilde{h}e^{\tilde{v}_s}}\right| \nonumber \\
&\leq& K_2, \nonumber \\
\label{uniform bound of h varphi s}
\end{eqnarray}

\noindent where $K_2$ is again a uniform constant. This is possible since $\tilde{h} \neq 0$ and for $s$ large enough, $\tilde{h}e^{\tilde{v}_s}$ and $\tilde{E}_s$ are smooth and uniformly bounded on $\cHol \times \Sigma$. We now possess sufficient estimates for the $L^p$ bound. For all $p\geq1$,

\begin{eqnarray}
\left|\int_\Sigma \tilde{h}e^{\tilde{\varphi}_s} dvol_\Sigma \right|^p &=& \left|\int_\Sigma \tilde{h}e^{\tilde{v}_s} dvol_\Sigma +\int_\Sigma \left(\tilde{h}e^{\tilde{\varphi}_s}-\tilde{h}e^{\tilde{v}_s}\right)vol_\Sigma \right|^p \nonumber \\
&\leq& \left(\int_\Sigma \left|\tilde{h}e^{\tilde{v}_s}\right|vol_\Sigma+\int_\Sigma \left|\tilde{h}e^{\tilde{\varphi}_s}-\tilde{h}e^{\tilde{v}_s}\right|dvol_\Sigma\right)^p \nonumber \\
&\leq& (K_1+K_2)^p \in L^1(\cHol),
\label{bound of integral h e varphi}
\end{eqnarray}

\noindent since $\cHol$ is compact. Finally, the integral

\[\int_\Sigma \left<\phi^\alpha,\phi^\beta\right>_H dvol_\Sigma\]

\noindent depends smoothly on variation of holomorphic maps and locations of zeros, and is uniformly bounded on $\cHol$. It is now clear that $|Z_{\alpha,\beta,s}|^p$ are uniformly bounded on $\cHol$, and therefore an $L^1$ function on $\cHol$. Dominated convergence theorem then applies to yield the conclusion of the proposition.

\end{proof}

The next desired convergence is

\begin{proposition}
For each $\alpha,\beta$, the smooth functions

\[Y_{\alpha,\beta,s}=-\int_\Sigma \left(\tilde{h}e^{\tilde{\varphi}_s}\right)^\alpha \tilde{u}_s^\beta dvol_\Sigma\]

\noindent converge to 0 as $s\to\infty$ in $L^p(\cHol)$ for all $p$.

\label{L p convergence of Y}

\end{proposition}

\begin{proof}

We begin by differentiating Kazdan-Warner equation as well as the approximated equation with respect to the $\beta^{th}$ coordinate of $\cHol$. The $\beta$-differentiation clearly commutes with $\Delta_\omega$ since they are defined on different spaces. We obtain the following linearizations:

\begin{equation}
\Delta_\omega \varphisb=-s^2\left(\tilde{h}^\alpha+\tilde{h}\varphisb\right)e^{\tilde{\varphi}_s},
\label{linearized Kazdan-Warner}
\end{equation}

\noindent and

\begin{equation}
\Delta_\omega \vsb=-s^2\left(\tilde{h}^\alpha+\tilde{h}\vsb\right)e^{\tilde{v}_s}+\tilde{E}_s^\alpha.
\label{linearized approximated Kazdan-Warner}
\end{equation}

The $\beta$-derivatives of $\tilde{v}_s$ and $\tilde{E}_s$ can be readily computed:

\begin{equation}
\vsb=\frac{1}{s^2}\frac{-s^2\tilde{h}}{\Delta_\omega\left(-\log(-\tilde{h})\right)-\tilde{c}(s)}\left[\frac{\Delta_\omega\left(-\log(-\tilde{h)}\right)-\tilde{c}(s)}{\tilde{h}}\right]^\beta,
\label{derivative of approximated solutions}
\end{equation}

\noindent and

\begin{equation}
\tilde{E}_s^\beta=\frac{1}{\tilde{c}(s)} \Delta_\omega\left[\frac{1}{\frac{\Delta_\omega\left(-\log(-\tilde{h)}\right)}{\tilde{c}(s)}-1}\right]\left[\Delta_\omega\left(-\log(-\tilde{h)}\right)\right]^\beta.
\label{derivative of error functions}
\end{equation}

\noindent One can readily verify that for $s$ large enough, the fact that $\tilde{h}$ is a smooth negative function on $\cHol\times\Sigma$ implies that both $\vsb$ and $\tilde{E}_s^\beta$ are $\frac{1}{s^2}$ times functions that are uniformly bounded on $\cHol\times\Sigma$, and therefore uniformly converge to 0 on $\cHol\times\Sigma$ as $s\to \infty$.

We then repeat the arguments of maximum principle as in the proof of Proposition \ref{L p convergence of Z} to the difference of equations \eqref{linearized Kazdan-Warner} and \eqref{linearized approximated Kazdan-Warner} to obtain uniformly decaying bounds for $\varphisb$ and $\vsb$. We find

\begin{equation}
\left|\varphisb-\vsb\right|\leq \sup_{\cHol\times\Sigma}\left|\tilde{E}_s^\beta\left(\frac{\tilde{h}^\beta\left(e^{\tilde{\varphi}_s}-e^{\tilde{v}_s}\right)+\left(\tilde{h}
e^{\tilde{\varphi}_s}-\tilde{h}e^{\tilde{v}_s}\right)\vsb}{\tilde{h}e^{\tilde{\varphi}_s}}\right)\right|.
\label{decaying of beta derivative of varphi}
\end{equation}

\noindent Since $\tilde{h}<0$ is smooth, $\tilde{h}^\beta$ is uniformly bounded on $\cHol\times\Sigma$. By \eqref{uniform bound of h varphi s}, we conclude that the fraction in the right hand side of \eqref{decaying of beta derivative of varphi} is uniformly bounded. Since $\tilde{E}_s^\beta$ decays to 0 uniformly on $\cHol\times\Sigma$ as $s\to\infty$, so does $\varphisb-\vsb$, which implies that $\varphisb \to 0$ uniformly as $s\to\infty$ since $\vsb$ does.

Recall the relation $\tilde{u}_s = \frac{1}{2}\tilde{\varphi}_s+\tilde{\psi}$, we then have

\begin{equation}
\tilde{u}_s^\beta \to \tilde{\psi}^\beta \hspace{1cm} \text{uniformly on }\cHol\times\Sigma \text{ as } s \to \infty,
\label{uniform decay of usb}
\end{equation}

\noindent where $\tilde{\psi}^\beta$ is the solution to

\[\Delta_\omega \tilde{\psi}^\beta = \left(\sqrt{-1}\Lambda F_{\tilde{H}}\right)^\beta.\]

\noindent Since $F_{\tilde{H}}$ is smooth on $\cHol\times\Sigma$, elliptic regularity ensures that $\tilde{\psi}^\beta$ is smooth and therefore uniformly bounded, and so is $\tilde{u}_s^\beta$. Combining this fact with equation (4.28) in \cite{L}, namely that $\left(\tilde{h}e^{\tilde{\varphi}_s}\right)^\alpha \to 0$ pointwise as $s\to \infty$. We have thus shown that $Y_{\alpha,\beta,s} \to 0$ as $s\to \infty$ pointwise on $\cHol$. It remains to construct a uniform bound to apply the dominated convergence theorem to establish convergence in $L^p(\cHol)$, which follows clearly from straightforward computations:

\begin{equation}
\left(\tilde{h}e^{\tilde{\varphi}_s}\right)^\alpha =\left(\tilde{h}^\alpha + \varphisa\right)e^{\tilde{\varphi}_s} \to \lim_{s \to \infty}\left(\tilde{h}^\alpha + \vsa\right)e^{\tilde{v}_s},
\label{bound for B}
\end{equation}

\noindent uniformly on $\cHol\times\Sigma$ as $s\to \infty$. Since every term in the limit is uniformly bounded over $s$, so is $\left(\tilde{h}e^{\tilde{\varphi}_s}\right)^\alpha$. It follows that $Y_{\alpha,\beta,s}$ are uniformly bounded on $\cHol$ and dominated convergence theorem is applicable.

\end{proof}

Finally we show

\begin{proposition}
For each $\alpha,\beta$, the smooth functions

\[X_{\alpha,\beta,s}= \int_\Sigma\frac{\left[\tilde{A}^\alpha + \left(\frac{\p\tilde{\varphi}_s^\alpha}{\p z}+2\frac{\p \tilde{\psi}^\alpha}{\p z}\right)dz\right]\wedge \bar{*}\left[\tilde{A}^\beta + \left(\frac{\p\tilde{\varphi}_s^\beta}{\p z}+2\frac{\p \tilde{\psi}^\beta}{\p z}\right)dz\right]}{2s^2}\]

\noindent converge to 0 as $s\to\infty$ in $L^p(\mathcal{H})$ for all $p$.

\label{L p convergence of X}
\end{proposition}

\begin{proof}

By Lemma 4.4 in \cite{L} and the smooth extension of $\tilde{h}$, $\frac{\p \tilde{\varphi}_s^\alpha}{\p z}$ and $\frac{\p \tilde{\varphi}_s^\beta}{\p z}$ are uniformly bounded. $\frac{\p \tilde{\psi}^\alpha}{\p z}$ and $\frac{\p \tilde{\psi}^\beta}{\p z}$ are smooth on $\cHol\times\Sigma$, independent of $s$, and therefore uniformly bounded. We need to, however, ensure that the integrals

\[I_{\alpha,\beta,s}^1:=\int_\Sigma \tilde{A}^\alpha \wedge \bar{*} \tilde{A}^\beta \]

\noindent and

\[I_{\alpha,\beta,s}^2:=\int_\Sigma \tilde{A}^\alpha  \wedge \bar{*}  \left(\frac{\p\tilde{\varphi}_s^\beta}{\p z}+2\frac{\p \tilde{\psi}^\beta}{\p z}\right)dz\]

\noindent are uniformly bounded on $\cHol$. The boundedness conditions ensure that the numerator of the integrand of $X_{\alpha,\beta,s}$ are uniformly bounded and therefore decays to 0 uniformly as $s \to \infty$.

Recall that, for $\tilde{\phi}\in\Hol$, $\tilde{A}(\tilde{\phi})$ is the unitary connection form with respect to the metric $\tilde{\phi}^*H_{FS}$. It is locally given by the (1,0) form

\[\tilde{A}(\tilde{\phi})=d'\left(\log \tilde{H}(\tilde{\phi})\right),\]

\noindent where

\[\tilde{H}(\tilde{\phi})= \frac{1}{\sum_{i=0}^k |\tilde{\phi_i}|^2},\]

\noindent is a smooth real function defined over a precompact coordinate neighborhood $U \subset \Sigma$. Using partition of unity to piece together these local expressions, we then have,

\begin{eqnarray}
I_{\alpha,\beta,s}^1 &=& \int_\Sigma \left[d'\left(\log \tilde{H}\right)\right]^\alpha \wedge \left[d''\left(\log \tilde{H}\right)\right]^\beta \nonumber \\
                     &=& \int_\Sigma d'\left(\log \tilde{H}\right)^\alpha \wedge d''\left(\log \tilde{H}\right)^\beta \nonumber \\
                     &=& -\int_\Sigma \left(\log \tilde{H}\right)^\alpha d'd''\left(\log \tilde{H}\right)^\beta \nonumber \\
                     &=& -\int_\Sigma \left(\log \tilde{H}\right)^\alpha \tilde{F}^\beta, \nonumber \\
\label{L 2 norm of initial connection forms}
\end{eqnarray}

\noindent where $\tilde{F}^\beta$ is the $\beta$-derivative of the curvature (1,1)-form of $\tilde{A}$. The indices $\alpha,\beta$ commute with integration and exterior derivatives since they are operators defined independently on different spaces. The third equality follows from Stoke's theorem. We observe that $\tilde{F}^\beta$ extends smoothly to a two form on $\Sigma$ depending smoothly on $\cHol$, and therefore $\int_\Sigma \tilde{F}^\beta$ is a uniformly bounded function on $\cHol$. $\log \tilde{H}$, however, might blow up near the boundary of $\Hol$. Nevertheless, we claim that the function

\begin{equation}
\int_\Sigma \left|\left(\log \tilde{H}\right)^\alpha\right| vol_\Sigma
\label{log integral}
\end{equation}

\noindent is uniformly bounded on $\cHol$. Indeed, the singularities of the function

\[\log \left(\frac{1}{\sum_{i=0}^k |\tilde{\phi_i}|^2}\right)\]

\noindent are precisely the common zeros of the $\tilde{\phi_i}$'s that develop near the boundary of $\cHol$. In the coordinate description of $\cHol$ given in \eqref{coordinate of Hol}, it corresponds to condition that for some $\alpha$, every $\tilde{\phi_i}$ vanishes at $w_\alpha$ with order $m^i_\alpha >0$. Let $m_\alpha$ be the minimum of these orders. Fix a open neighborhood $U\subset \Sigma$ around $w_\alpha$ such that for all $i$, $|\tilde{\phi_i}|=|z-w_\alpha|^{m_\alpha}f_i$, where $f_i$ is a nonvanishing smooth function on $U$. We then rewrite

\[\log \tilde{H}=\log \left(\frac{1}{\sum_{i=0}^k |\tilde{\phi_i}|^2}\right)=m_\alpha \log \left(\frac{1}{|z-w_\alpha|^2}\right)+\tilde{G},\]

\noindent where $\tilde{G}$ is a smooth function on $\cHol\times U$. Differentiating with respect to $w_\alpha$ and taking the norm, we have

\begin{equation}
\left|\left(\log\tilde{H}\right)^\alpha\right|=\frac{\Big|1-\tilde{G}^\alpha(z-w_\alpha)\Big|}{|z-w_\alpha|},
\label{alpha derivative of log H}
\end{equation}

\noindent which is integrable over $U$, using polar coordinate, and the integral depends smoothly on $w_\alpha$, the location of common zeros for $\tilde{\phi_i}$'s. We conclude that there is a constant $K_U$ such that

\begin{equation}
\int_U \left|\left(\log \tilde{H}\right)^\alpha\right| dvol_\Sigma \leq K_U
\label{integral of log H derivative}
\end{equation}

\noindent holds on all $\cHol$. Since $\Sigma$ is compact, we conclude that

\[\int_\Sigma \left|\left(\log \tilde{H}\right)^\alpha\right| dvol_\Sigma\]

\noindent is uniformly bounded over $\cHol$. This proves the claim.

From the claim and the uniform bound of $\int_\Sigma \tilde{F}^\beta$, we conclude that for all $\alpha,\beta$, $I_{\alpha,\beta,s}^1$ are uniformly bounded, over $s$, on $\cHol\times\Sigma$. In fact, since

\[ \left(\frac{\p\tilde{\varphi}_s^\beta}{\p z}+2\frac{\p \tilde{\psi}^\beta}{\p z}\right)\]

\noindent are uniformly bounded on $\cHol\times\Sigma$, estimate on $\int_\Sigma \left|\left(\log \tilde{H}\right)^\alpha \right| dvol_\Sigma$ also shows that $I_{\alpha,\beta,s}^2$ are uniformly bounded for all $\alpha,\beta,s$.

\end{proof}

The main theorem is now an immediate consequence of the analytic results we have established.

\begin{proof}(\emph{of the Main Theorem \ref{Main Theorem}})

In terms of local representations of metrics on coordinates $(w_1,\ldots,w_m)$ on $\mathcal{U}$, the volume forms for each $s$ is

\[dvol_s^* = \sqrt{\det(g_{\alpha,\beta,s}^*)}dw_1 \wedge \cdots \wedge dw_m.\]

\noindent From Propositions \ref{L p convergence of Z}-\ref{L p convergence of X}, it follows that $\det(g_{\alpha,\beta,s})^*$ converge to $\det(g_{{L^2},\alpha,\beta})$ in all $L^p(\cHol)$, and therefore $L^p(\Hol)$.

Covering $\Hol$ by local coordinate patches $\{\mathcal{U}\}$, there is a smooth partition of unity $\{\psi_\mathcal{U}\}$ subordinate to this covering. The $L^2$ volume of $\Hol$ is then

\begin{eqnarray}
\text{Vol }\Hol &=& \int_{\Hol} dvol_{L^2} \nonumber \\
                &=&\sum_\mathcal{U} \psi_\mathcal{U}\int_\mathcal{U}  dvol_{L^2} \nonumber \\
                &=&\sum_\mathcal{U}\psi_\mathcal{U} \int_\mathcal{U} \lim_{s\to\infty} dvol_s^*  \nonumber \\
                &=&\sum_\mathcal{U}\psi_\mathcal{U} \lim_{s\to\infty}\int_\mathcal{U}  dvol_s^*  \nonumber \\
                &=&\sum_\mathcal{U} \psi_\mathcal{U}\lim_{s\to\infty}\int_{\Phi_s^*{\mathcal{U}}}  dvol_s^*  \nonumber \\
                &=&\lim_{s\to\infty}\int_{\nu_{k+1,0}(s)}  dvol_s . \nonumber \\
                \label{limiting L 2 metric}
\end{eqnarray}

\noindent \eqref{limiting L 2 metric} then allows us to take the limit of the volume of $\nu_{k+1,0}(s)$ given by \eqref{volume of moduli space nonvanishing} as $s \to \infty$. Since the volume of $\Sigma$ is normalized to be $1$, we have proved the validity of the formula.

\end{proof}

One might notices that the extension of the metrics and curvature forms over the boundary of $\Hol$ lowers the topological degree of the pullback line bundle $\tilde{\phi}^* \mathcal{O}(1)$. In fact, we have experimented the bubbling phenomenon in the elementary settings. The author is eager to explore generalization of this result to more general cases, such as the one posed in \cite{Ba1}.

\renewcommand{\abstractname}{Acknowledgements}
\begin{abstract}
This work is supported by the Ministry of Science and Technology of the Republic of China (Taiwan) grant 103-2115-M-006-015-MY2.
\end{abstract}

\end{document}